\documentclass[12pt]{article}

\usepackage{amssymb,amsfonts,amsmath,amsthm,amscd}
\usepackage{fullpage}

\DeclareMathOperator{\Th}{Th}
\DeclareMathOperator{\pd}{pd}

\def\primes{{\cal P}}
\def\Enn{\mathbb{N}}
\def\Zee{\mathbb{Z}}
\def\Que{\mathbb{Q}}
\def\Quep{\Que^{\geq 0}}
\def\lmod#1 #2{#1\ ({\rm mod}\ #2)}
\def\quo{{\rm quo}}
\def\pref{{\rm pref}}
\newcommand{\monus}{\stackrel{{}^{\scriptstyle .}}{\smash{-}}}
\def\rela{\vartriangleleft}
\def\divides{{ \  | \  }}
\def\nodiv{{\, |\kern-3.9pt/}}

\theoremstyle{plain}
\newtheorem{theorem}{Theorem}
\newtheorem{corollary}[theorem]{Corollary}
\newtheorem{lemma}[theorem]{Lemma}

\theoremstyle{definition}
\newtheorem{definition}[theorem]{Definition}
\newtheorem{example}[theorem]{Example}
\newtheorem{conjecture}[theorem]{Conjecture}

\theoremstyle{remark}
\newtheorem{remark}[theorem]{Remark}

\begin{document}

\title{Automatic Sets of Rational Numbers}

\author{Eric Rowland\\
Universit\'e de Li\`ege \\
D\'epartement de Math\'ematiques \\
Grande Traverse 12 (B37) \\
4000 Li\`ege \\
Belgium \\
{\tt erowland@ulg.ac.be} \\
\ \\
Jeffrey Shallit\footnote{Corresponding author.} \\
School of Computer Science \\
University of Waterloo \\
Waterloo, ON  N2L 3G1 \\
Canada \\
{\tt shallit@cs.uwaterloo.ca} \\
}

\maketitle

\begin{abstract}
The notion of a $k$-automatic set of integers is well-studied.
We develop a new notion --- the $k$-automatic set of rational numbers --- and
prove basic properties of these sets, including closure properties and
decidability.  
\end{abstract}

\section{Introduction}
\label{intro}

Let $k$ be an integer $\geq 2$, and let $\Enn = \lbrace 0, 1, 2, \ldots
\rbrace$ denote the set of non-negative integers.  Let
$\Sigma_k = \lbrace 0, 1, \ldots, k-1 \rbrace$ be an alphabet of $k$ letters.
Given a word $w = a_1 a_2 \cdots a_t \in \Sigma_k^*$, we let
$[w]_k$ denote the integer that it represents in
base $k$; namely,
\begin{equation}
[w]_k = \sum_{1 \leq i \leq t} a_i k^{t-i},
\label{wk}
\end{equation}
where (as usual) an empty sum is equal to $0$.
For example, $[0101011]_2 = 43$.  

Note that in this framework, every element of $\Enn$ has infinitely many
distinct representations as words, each having a certain number of
\textit{leading zeroes}.  Among all such representations, the one
with no leading zeroes is called the \textit{canonical representation};
it is an element of $C_k := \lbrace \epsilon \rbrace \ \cup \ 
(\Sigma_k - \lbrace 0 \rbrace)\Sigma_k^*$.    For an integer $n \geq 0$,
we let $(n)_k$ denote its canonical representation.  Thus, for
example, $(43)_2 = 101011$.  Note that
$\epsilon$ is the canonical representation of $0$.

Given a language $L \subseteq \Sigma_k^*$, we can define the set of
integers it represents, as follows:
\begin{equation}
[L]_k = \lbrace [w]_k \, : \, w \in L \rbrace .
\end{equation}

We now recall a well-studied concept, that of $k$-automatic set
(see, e.g.,
\cite{Cobham:1969,Cobham:1972,Allouche&Shallit:2003}):

\begin{definition}
We say that a set $S \subseteq \Enn$ is \textit{$k$-automatic} if there exists
a regular language $L \subseteq \Sigma_k^*$ such that $S = [L]_k$.
\label{def1}
\end{definition}

Many properties of these sets are known.  For example, it is possible
to state an equivalent definition involving only canonical
representations:

\begin{definition}
A set $S \subseteq \Enn$ is \textit{$k$-automatic} if the language
$$(S)_k := \lbrace (n)_k \, : \, n \in S \rbrace$$
is regular.
\label{def2}
\end{definition}

To see the equivalence of Definitions~\ref{def1} and \ref{def2},
note that if $L$ is a regular language, then so
is the language $L'$ obtained by removing all leading zeroes from
each word in $L$.

A slightly more general concept is that of $k$-automatic \textit{sequence}.
Let $\Delta$ be a finite alphabet.  Then a sequence (or infinite word)
$(a_n)_{n \geq 0}$ over $\Delta$ is said to be $k$-automatic if, for every
$c \in \Delta$, the set of \textit{fibers}
$ F_c = \lbrace n \in \Enn \, : \, a_n = c \rbrace$
is a $k$-automatic set of natural numbers.
Again, this class of sequences has been
widely studied \cite{Cobham:1969,Cobham:1972,Allouche&Shallit:2003}.
The following result is well-known \cite{Eilenberg:1974}:

\begin{theorem}
The sequence $(a_n)_{n \geq 0}$ is $k$-automatic if and only if its
$k$-kernel,
the set of its subsequences
$ K = \lbrace (a_{k^e n + f})_{n \geq 0} \, : \, e \geq 0,\ 0 \leq f < k^e
	\rbrace,$
is finite.
\label{kernel}
\end{theorem}

In previous papers \cite{Shallit:2011,Schaeffer&Shallit:2012},
the second author extended
the notion of $k$-automatic sets over
$\Enn$ to subsets of
$\Quep$, the non-negative rational numbers.  The motivation was to
study the ``critical exponent'' of automatic sequences.
In this paper,
we will obtain some basic results about this new class.
Our principal results are Theorem~\ref{integers} (characterizing
those $k$-automatic sets of rationals consisting entirely of integers),
Theorem~\ref{inter} and Corollary~\ref{notclosed} (showing that the class of $k$-automatic sets of
rationals is not closed under intersection or complement),
Theorem~\ref{infinitec}
(showing that it is decidable if a $k$-automatic set of rationals is infinite),
and Theorem~\ref{algo} (showing that it is decidable if a $k$-automatic
set of rationals equals $\Enn$).

The class of sets we study has some similarity to another class studied by
Even \cite{Even:1964} and
Hartmanis and Stearns \cite{Hartmanis&Stearns:1967}; their class 
corresponds to the topological closure of a small subclass of
our $k$-automatic sets, in which
the possible denominators are restricted to powers of $k$.

Yet another model of automata accepting real numbers was studied
in \cite{Adamczewski&Bell:2011,Boigelot&Brusten&Bruyere:2010,Boigelot&Brusten:2009,Boigelot&Brusten&Leroux:2009}.  In this model real numbers are represented
by their (possibly infinite) base-$k$ expansions, and the model of automaton
used is a nondeterministic B\"uchi automaton.  However, even when restricted
to rational numbers, this model does not define the same class of sets,
as we will show below in Corollary~\ref{notbuchi}.

A preliminary version of this paper appeared in \cite{Rowland&Shallit:2012}.

\section{Representing rational numbers}

A natural representation for the non-negative rational number $p/q$ is the
pair $(p,q)$ with $q \not= 0$.  Of course, this representation has the
drawback that every element of $\Quep$ has infinitely many representations,
each of the form $(jp/d, jq/d)$ for some $j \geq 1$, where
$d = \gcd(p,q)$. 

We might try to ensure uniqueness of representations by considering
only ``reduced'' representations (those in ``lowest terms''),
which amounts to representing $p/q$ by the pair $(p/d, q/d)$ where
$d = \gcd(p,q)$.  In other words, the only valid pairs are $(p,q)$ with
$\gcd(p,q) = 1$. However, the condition $\gcd(p,q) = 1$ cannot be checked,
in general,
by finite or even pushdown automata --- see Remark~\ref{rem1} below
--- and it is not currently known if it is decidable whether
a given regular language consists entirely of 
reduced representations (see Section~\ref{open}).
Furthermore, insisting on only reduced representations
means that some ``reasonable'' sets of rationals, such 
as $\lbrace (k^m-1)/(k^n - 1) \, : \, m, n \geq 1 \rbrace$ (see
Corollary~\ref{lowestterms}), have no representation as a regular
language.  For these reasons,
\textbf{we allow the rational number
$p/q$ to be represented by \textit{any} pair of
non-negative integers $(p',q')$ with $p/q = p'/q'$}.

Next, we need to see how to represent a pair of integers as a word over
a finite alphabet.  Here, we follow the ideas of Salon 
\cite{Salon:1986,Salon:1987,Salon:1989}.
Consider the alphabet $\Sigma_k^2$.  A finite word $w$ over $\Sigma_k^2$ can
be considered as
a sequence of pairs $w = [a_1, b_1][a_2,b_2] \cdots [a_n, b_n]$ where
$a_i, b_i \in \Sigma_k$ for $1 \leq i \leq n$.  We can now define
the projection maps $\pi_1$, $\pi_2$ from $(\Sigma_k^2)^*$ to
$\Sigma_k^*$, as follows:
$$ \pi_1 (w) = a_1 a_2 \cdots a_n; \quad\quad\quad 
\pi_2 (w) = b_1 b_2 \cdots b_n . $$
Then we define $[w]_k = ([\pi_1(w)]_k, [\pi_2(w)]_k)$.
Thus, for example, if 
$$w = [0,0][1,0] [0,1] [1,0] [0,0] [1,1] [1,0],$$
then $[w]_2 = (43, 18)$.    
We also define $\times$, which allows us to
join two words $w, x \in \Sigma_k^*$ of the same length to create a single
word $w \times x \in (\Sigma_k^2)^*$ whose $\pi_1$ projection is $w$
and $\pi_2$ projection is $x$.

In this framework, every pair of integers $(p,q)$ again
has infinitely many distinct representations,
arising from padding on the left by leading pairs of zeroes, that is, by
$[0,0]$.  Among all such representations, the \textit{canonical
representation} is the one having no leading pairs of zeroes.
We write it as $(p,q)_k$.   For example,
$(43,18)_2 =  [1,0] [0,1] [1,0] [0,0] [1,1] [1,0]$.

We now state the fundamental definitions of this paper:

\begin{definition}
Given a word $w \in (\Sigma_k^2)^*$ with $[\pi_2(w)]_k \not= 0$, we define
$$\quo_k(w) := {{[\pi_1(w)]_k} \over {[\pi_2(w)]_k}}.$$
If $[\pi_2(w)]_k = 0$, then $\quo_k(w)$ is not defined.
Further, if $[\pi_2(w)]_k \not= 0$ for all $w \in L$, then
$ \quo_k(L) := \lbrace \quo_k(w) \, : \, w \in L \rbrace$.
A set of rational numbers $S \subseteq \Quep$ is 
\emph{$k$-automatic} if there exists a regular language 
$L \subseteq (\Sigma_k^2)^*$
such that $S = \quo_k(L)$.  
\label{def4}
\end{definition}

We reiterate that for a rational number 
$\alpha$ to be in $\quo_k(L)$, only a single,
possibly non-reduced, representation of $\alpha$ need be in $L$.
Furthermore, $L$ may contain
multiple representations for $\alpha$ in two different ways:
$L$ could contain non-canonical
representations that begin with leading zeroes, and $L$ could contain
``unreduced'' representations $(p,q)$ where $\gcd(p,q) > 1$.

We could have adopted a different definition, by insisting that every
rational in $\quo_k (L)$ must have every possible representation in
$L$.  However, this would have the unpleasant consequence that some
very simple subsets of $\Quep$ would not have a regular
representation.  For example, the language of all representations of
$\Enn$ is given by $L_d = \{ (p,q)_k \, : \, q \divides p \}$, which is
not even context-free; see Remark~\ref{rem1}.

Another issue is the following:
given a set $S \subseteq \Quep$, if $S$ contains a non-integer, then
by calling $S$ $k$-automatic, it is clear that we intend this to mean
that $S$ is $k$-automatic in the sense of Definition~\ref{def4}.
But what if $S \subseteq \Enn$?  
Then calling it ``automatic'' might mean {\it either\/} automatic
in the usual sense, as in Definition~\ref{def1}, {\it or\/} in the
extended sense introduced in this section, treating $S$ as a subset
of $\Quep$.  In Theorem~\ref{integers} we will see that these two 
interpretations actually \textit{coincide} for subsets of $\Enn$,
but in order to prove this, we need some notation to distinguish between
the two types of representations.   So, from now on,
by $(\Enn,k)$-automatic we
mean the interpretation in Definition~\ref{def1}, and by
$(\Quep,k)$-automatic
we mean the interpretation in Definition~\ref{def4}.

Yet another issue is the order of the representations.
So far we have only considered representations where the leftmost digit
is the most significant digit.  However, sometimes it is simpler to deal with
\textit{reversed representations} where the leftmost digit is the
least significant digit.  In other words, sometimes it is easier to
deal with the reversed word $w^R$ and reversed language $L^R$ instead
of $w$ and $L$, respectively.  Since the regular languages are (effectively)
closed
under reversal, for most of our results it does not matter which
representation we choose, and we omit extended discussion of this point.

We use the following notation for intervals:  $I[\alpha, \beta]$ denotes
the closed interval 
$$\lbrace x \, : \, \alpha \leq x \leq \beta \rbrace,$$
and similarly for open- and half-open intervals.

\section{Examples}

To build intuition, we give some examples of $k$-automatic sets of rationals.

\begin{example}
Consider the regular language 
$$L_0 = \lbrace w \in (\Sigma_k^2)^* \, : \, \pi_1(w) \in C_k \ \cup \ 
\lbrace 0 \rbrace
\text{ and } [\pi_2(w)]_k = 1 \rbrace.$$
Then $\quo_k (L_0) = \Enn$, as $L_0$ contains words with arbitrary numerator,
but denominator equal to $1$.
\label{example2}
\end{example}

\begin{example}
Let $k = 2$, and consider the regular language $L_1$ defined by
the regular expression 
$ A^* \lbrace [0,1], [1,1] \rbrace A^*,$
where $A = \lbrace [0,0],[0,1],[1,0],[1,1] \rbrace$.
This regular expression specifies all pairs of integers where the
second component has at least one nonzero digit --- the point being to avoid 
division by $0$.  Then $\quo_k(L_1) = \Quep$, the
set of all non-negative rational numbers. In fact all possible representations
of all rational numbers are included in $L_1$.
\end{example}

\begin{example}
Consider the regular language 
$$L_2 = 
\lbrace w \in (\Sigma_k^2)^* \ : \ \pi_2 (w) \in 0^* 1^+ 0^* \rbrace. $$
Then we claim that  $\quo_k(L_2) = \Quep$.    To see this,
consider an arbitrary non-negative rational number $p/q$ with $q \geq 1$.
Let $i$ be the least non-negative integer  such that
$\gcd(k^i, q) = \gcd(k^{i+1}, q)$. Let $d = \gcd(k^i, q)$ and
write $q = d q'$; note that $\gcd(k, q') = 1$ and hence
$\gcd(k, (k-1)q') = 1$.  
By the Fermat-Euler theorem, there is an
integer $j \geq 1$ such that $k^j \equiv \lmod{1} {(k-1)q'}$.
(For example, we can take $j = \varphi( (k-1) q' )$, where
$\varphi$ is Euler's totient function.)
Define $d' = k^i/d$
and $t = {{k^j - 1} \over {(k-1)q'}} $; then
$$ {p \over q} = {{d'tp} \over {d'tq}} = {{d'tp} \over {k^i {{k^j-1}\over{k-1}}}} ,$$
which expresses $p/q$ as a number with denominator of the required form,
with base-$k$ representation $1^j \, 0^i$.  In Theorems~\ref{dd1} and
\ref{dd2} below we will
see that $L_2$ achieves the minimum subword complexity of $\pi_2(L)$
over all regular languages $L$ representing $\Quep$.
\label{six}
\end{example}

\begin{example}
For a word $x$ and letter $a$ let $|x|_a$ denote the number
of occurrences of $a$ in $x$.
Consider the regular language
$$ L_3 = \lbrace w \in (\Sigma_2^2) \ : \ |\pi_1 (w)|_1 
\equiv \lmod{0} {2} \text{ and }
|\pi_2 (w)|_1 \equiv \lmod{1} {2} \rbrace.$$
Then it follows from a result of Schmid \cite{Schmid:1984}
that $\quo_2(L_3) = \Quep - \lbrace 2^n \ : \ n \in \Zee \rbrace$.
\label{three}
\end{example}

\begin{example}
Let $k = 3$, and consider the regular language $L_4$ defined by the
regular expression $[0,1] \lbrace [0,0],[2,0] \rbrace^*$.
Then $\quo_k(L_4)$ is the \textit{$3$-adic Cantor set}, the set of
all rational numbers in the ``middle-thirds'' Cantor set 
with denominators a power of $3$ \cite{Cantor:1883}.
\end{example}

\begin{example}
Let $k = 2$, and consider the regular language $L_5$ defined
by the regular expression 
$ [0,1] \lbrace [0,0], [0,1] \rbrace^* \lbrace [1,0],[1,1] \rbrace.$
Then the numerator encodes the integer $1$, while the denominator
encodes all integers that start with $1$.  Hence
$\quo_k(L_5) = \lbrace {1 \over n} \, : \, n \geq 1 \rbrace$.
\end{example}

\begin{example}
Let $k = 4$, and consider the set $S = \lbrace 0, 1, 3, 4, 5, 11, 12, 13,
\ldots \rbrace$ of all non-negative integers
that can be represented using only the digits $0, 1, -1$ in base $4$.
Consider the set $L_6 = \lbrace (p,q)_4 \, : \, p, q \in S \rbrace$. It is
not hard to see that $L_6$ is $(\Quep, 4)$-automatic.
The main result in \cite{Loxton&vanderPoorten:1987} can be
phrased as follows:  $\quo_4 (L_6)$ contains every odd integer.
In fact, an integer $t$ is in $\quo_4 (L_6)$ if and only if
the exponent of the largest power of $2$ dividing $t$ is even.
\end{example}

\begin{example}
Let $K, L$ be arbitrary regular languages over the alphabet $\Sigma_k^2$.
Note that $\quo_k (K \, \cup \, L) = \quo_k(K) \, \cup \, \quo_k(L)$
but the analogous identity involving intersection need not hold.  As
an example, consider $K = \lbrace [2,1] \rbrace$ and
$L = \lbrace [4,2] \rbrace$.  Then $\quo_{10} (K \, \cap \, L)
= \emptyset \not= \lbrace 2 \rbrace = \quo_{10} (K) \, \cap \, \quo_{10}(L)$.
\end{example}

\section{Basic results}

In this section we obtain some basic results about automatic sets of
rationals.

\begin{theorem}
Let $r, s$ be integers $\geq 1$.
Then $S$ is a $k^r$-automatic set of rational numbers 
if and only if $S$ is $k^s$-automatic.
\end{theorem}

\begin{proof}
Follows easily from the same result for automatic sequences
\cite{Cobham:1969,Cobham:1972}.
\end{proof}

Next we state a useful result from \cite{Shallit:2011,Schaeffer&Shallit:2012}:

\begin{lemma}
Let $\beta$ be a non-negative real number and define the
languages
$$L_{\leq \beta} = \lbrace x \in (\Sigma_k^2)^* \ : \ \quo_k(x) \leq \beta \rbrace,$$
and analogously for the relations $<, =, \geq, >, \not=$.
\begin{itemize}
\item[(a)] If $\beta$ is a rational number, then the language
$L_{\leq \beta}$ (resp., $L_{<\beta}$, $L_{=\beta}$, $L_{\geq \beta}$,
$L_{>\beta}$, $L_{\not=\beta}$) is regular.

\item[(b)] If $L_{\leq \beta}$ (resp., $L_{<\beta}$, 
$L_{\geq \beta}$, $L_{>\beta}$) is regular,
then $\beta$ is a rational number.
\end{itemize}
\label{lem1}
\end{lemma}

Suppose $S$ is a set of real numbers, and $\alpha$ is a real number.  We
introduce the following notation:
\begin{align*}
S+\alpha & :=   \lbrace x+\alpha \, : \, x \in S \rbrace  \\
S \monus \alpha  & :=  \lbrace \max(x-\alpha,0) \, : \, x \in S \rbrace \\
\alpha \monus S & :=  \lbrace \max(\alpha-x, 0) \, : \, x \in S \rbrace \\ 
\alpha S  & :=   \lbrace \alpha x \, : \, x \in S \rbrace.
\end{align*}

\begin{theorem}
The class of $k$-automatic sets of rational numbers
is closed under the following operations:
\begin{itemize}
\item[(i)] union;
\item[(ii)] $S \rightarrow S + \alpha$ for 
$\alpha \in \Quep$;
\item[(iii)] $S \rightarrow S \monus \alpha$ for 
$\alpha \in \Quep$;
\item[(iv)] $S \rightarrow \alpha \monus S$ for 
$\alpha \in \Quep$;
\item[(v)] $S \rightarrow \alpha S$ for 
$\alpha \in \Quep$;
\item[(vi)]
$S \rightarrow \lbrace 1/x \, : \, x \in S \setminus \lbrace 0 \rbrace 
\rbrace$.
\end{itemize}
\end{theorem}

\begin{proof}
We prove only item (ii), with the others being similar.
We will use the reversed representation, with the least significant
digit appearing first.
Write $\alpha = p/q$.  Let $M$ be a DFA with $\quo_k(L(M)) = S$.
To accept $S + \alpha$, on input
a base-$k$ representation of $x = p'/q'$, we transduce the numerator to
$p'q - pq'$ and the denominator to $qq'$
(hence effectively computing a representation for $x - \alpha$), and 
simultaneously simulate $M$ on this input, digit-by-digit, accepting
if $M$ accepts.  
\end{proof}

The next theorem shows that if $p/q$ has a representation in a regular
language, then it has a representation where the numerator and
denominator are not too large.  Here we are using the least-significant-digit
first representation.

\begin{theorem}
Let $L \subseteq (\Sigma_k^2)^*$ be a regular language, accepted
by an NFA $M$ with $n$ states.  
If $p/q \in \quo_k (L)$, there exists 
$w \in L$ with
$\pi_1 (w) = p'$, 
$\pi_2 (w) = q'$,
$p'/q' = p/q$, and
$p', q' < k^{pqn}$.
\end{theorem}

\begin{proof}
On input $w$ representing the pair $(p', q')$ 
we compute $qp'$ and $pq'$ simultaneously, digit-by-digit, and test
if they are equal.  To carry out multiplication by $q$, we need to
keep track of carries, which could be as large as $q-1$, and similarly
for $p$.  An NFA $M'$ to do this can be built using
triples $(i,j,q_l)$, where $i$ is a carry in the multiplication by
$p$, $j$ is a carry in the multiplication by $q$, and $q_l$ is the
state in $M$ arising from processing a prefix of $w$.
Thus $M'$ has $pqn$ states.  If $M'$ accepts any word, then it accepts
a word of length at most $pqn-1$.  From this the inequality follows.
\end{proof}

We now state one of our main results.

\begin{theorem}
Let $S \subseteq \Enn$.  Then $S$ is $(\Enn,k)$-automatic if and only if
it is $(\Quep,k)$-automatic.
\label{integers}
\end{theorem}

The proof requires a number of preliminary results.
First, we introduce some terminology and notation.
We say a set $S \subseteq \Enn$ is {\it ultimately periodic} if there exist
integers $n_0 \geq 0, p \geq 1$ such that 
$n \in S \iff n+p \in S$, provided $n \geq n_0$.
In particular, every finite set is ultimately periodic.

We let $\primes = \lbrace 2,3,5, \ldots \rbrace$ denote the set
of prime numbers.  Given a positive integer $n$, we let
$\pd(n)$ denote the set of its prime divisors.  For example,
$\pd(60) = \lbrace 2, 3, 5 \rbrace$.  Given a subset
$D \subseteq \primes$, we let $\pi(D) = \lbrace n \geq 1 \, : \, \pd(n) \subseteq D \rbrace$, the set of all integers that can be factored completely using
only the primes in $D$.  
Finally, recall that
$\nu_k (n)$ denotes the exponent of the largest power of $k$ dividing $n$.

First, we prove two useful lemmas.

\begin{lemma}
Let $S \subseteq \Enn- \lbrace 0 \rbrace$.  Then the following
are equivalent:
\begin{itemize}
\item[(a)]  There exist an integer $n \geq 0$, and $n$ integers $g_i \geq 1$,
$1 \leq i \leq n$, and $n$ ultimately periodic subsets $W_i \subseteq \Enn$,
$1 \leq i \leq n$, such that
$$ S = \bigcup_{1 \leq i \leq n} g_i \lbrace k^j \ : \ j \in W_i \rbrace;$$

\item[(b)]  There exist an integer $m \geq 0$, and $m$ integers $f_i$
with $k \nodiv f_i$, $1 \leq i \leq m$, and $m$ ultimately periodic
subsets $V_i \subseteq \Enn$, $1 \leq i \leq m$, such that
\begin{equation}
S = \bigcup_{1 \leq i \leq m} f_i \lbrace k^j \ : \ j \in V_i \rbrace,
\label{b1}
\end{equation}
and the union is disjoint.

\item[(c)]  Define $F = \lbrace s/{k^{\nu_k (s)}} \ : \ s \in S \rbrace$.
The set $F$ is finite, and for all $f \in F$, the set 
$U_f = \lbrace j \ : \ k^j f \in S \rbrace$ is ultimately periodic.
\end{itemize}
\label{kfinite}
\end{lemma}

\begin{proof}
(a) $\implies$ (b):  For each $g_i$ define $x_i = \nu_k (g_i)$ and
$f_i = g_i / k^{x_i}$.  Note that $k \nodiv f_i$.  Then
\begin{align*}
S &= \bigcup_{1 \leq i \leq n} g_i \lbrace k^j \ : \ j \in W_i \rbrace \\
&= \bigcup_{1 \leq i \leq n} f_i k^{x_i} \lbrace k^j \ : \ j \in W_i \rbrace \\
&= \bigcup_{1 \leq i \leq n} f_i \lbrace k^{x_i + j} \ : \ j \in W_i \rbrace \\
&= \bigcup_{1 \leq i \leq n} f_i \lbrace k^j \ : \ j \in W'_i \rbrace,
\end{align*}
where $W'_i = x_i + W_i$.  Note that each $W'_i$ is ultimately 
periodic.
If any of the $f_i$ coincide, we take the union of the corresponding $W'_i$
and call it $V_i$.  Since the union of a finite number of ultimately
periodic sets is still ultimately periodic (see, e.g., \cite{Matos:1994}), we can choose
a subset of the indices $i$ so that each $f_i$ appears once, expressing
$S$ as  
$$ \bigcup_{1 \leq i \leq m} f_i \lbrace k^j \ : \ j \in V_i \rbrace$$
for some $m \leq n$.  The union is now disjoint, for if
(say) $f_1 k^j = f_2 k^{j'}$ for $j \in V_1$ and $j' \in V_2$, then 
since $k \nodiv f_1, f_2$ we must have $j = j'$ and so $f_1 = f_2$,
which is a contradiction.

\bigskip
\noindent (b) $\implies$ (c):

Let $s \in S$.  Then $s = f_i k^j$ for some $j$.  Since
$k \nodiv f_i$, we have $s/{k^{\nu_k (s)}} = f_i$.  Thus
$F = \lbrace f_i \ : \ 1 \leq i \leq m \rbrace$ and hence is finite.
By the disjointness of the union (\ref{b1}) (and the fact that $k \nodiv f_i$)
we have $k^j f_i \in S \iff j \in V_i$.  So $U_{f_i} = V_i$ and hence
$U_{f_i}$ is ultimately periodic.

\bigskip
\noindent (c) $\implies$ (a):

Let $g_1, g_2, \ldots, g_n$ be distinct elements such that
$F = \lbrace g_1, g_2, \ldots, g_n \rbrace$.  Take $W_i = U_{g_i}$.
\end{proof}

If any of the conditions (a)--(c) above hold, we say that the set $S$ is
{\it $k$-finite}.

\begin{lemma}
Let $D$ be a finite set of prime numbers, and let $S \subseteq \pi(D)$.
Let $s_1, s_2, \ldots $ be an infinite sequence of (not necessarily
distinct) elements of $S$.  Then there is an infinite increasing
sequence of indices $i_1 < i_2 < \cdots$ such that
$s_{i_1} \divides s_{i_2} \divides \cdots $.  
\label{dick}
\end{lemma}

\begin{proof}
Case 1:  The sequence $(s_i)$ is bounded.
In this case infinitely many of the $s_i$ are the same, so we can
take the indices to correspond to these $s_i$.

Case 2:  The sequence $(s_i)$ is unbounded.   In this case we prove
the result by induction on $|D|$.
If $|D| = 1$, then we can choose a strictly increasing subsequence
of the $(s_i)$; since all are powers of some prime $p$, this subsequence
has the desired property.

Now suppose the result is true for all sets $D$ of cardinality $t-1$.  We
prove it for $|D| = t$.
Since only
$t$ distinct primes figure in the factorization of the $s_i$,
some prime, say $p$, must appear with unbounded exponent in the $(s_i)$. 
So there is some
subsequence of $(s_i)$, say $(t_i)$, with strictly increasing exponents
of $p$.  Now consider the infinite sequence $(u_i)$ given by
$u_i = t_i/p^{\nu_p (t_i)}$.  Each $u_i$ has a prime factorization
in terms of the primes in $D - \lbrace p \rbrace$, so by induction
there is an infinite increasing sequence of indices 
$i_1, i_2, \ldots$ such that $u_{i_1} \divides u_{i_2} \divides \cdots$.
Then $p^{\nu_p (t_{i_1})} u_{i_1} \divides
p^{\nu_p (t_{i_2})} u_{i_2} \divides \cdots$, which corresponds
to an infinite increasing sequence of indices of the original
sequence $(s_i)$.
\end{proof}

We now state an essential part of the proof, which is of independent 
interest.

\begin{theorem}
Let $D \subseteq \primes$ be a finite set of prime numbers, and
let $S \subseteq \pi(D)$.  Then $S$ is $k$-automatic if and only if it is
$k$-finite.
\label{dkf}
\end{theorem}

\begin{proof}
$\Longleftarrow$:  If $S$ is $k$-finite, then
by Lemma~\ref{kfinite}, we can write it as the disjoint finite union
$$S = \bigcup_{1 \leq i \leq m} f_i \lbrace k^j \ : \ j \in V_i \rbrace ,$$
where each $V_i$ is an ultimately periodic set of integers.
For each $i$, the set  
$(f_i)_k \lbrace 0^t \, : \, t \in V_i \rbrace$ of base-$k$ representations of
$f_i \lbrace k^j \ : \ j \in V_i \rbrace$ is a regular language.
It follows that $S$ is $k$-automatic.

\bigskip

$\Longrightarrow$:  Suppose $S$ is $k$-automatic.
Define $F = \lbrace s/{k^{\nu_k (s)}} \ : \ s \in S \rbrace$.
Suppose $F$ is infinite.  
Clearly no element of $F$ is divisible by $k$.
Therefore we can write
$S$ as the disjoint union $\bigcup_{i \geq 0} k^i H_i$, where 
$H_i := \lbrace f \in F \, : \, k^i f \in S \rbrace$.
Then there are two possibilities:  either (a) the sets
$H_i$ are finite for all $i \geq 0$, or (b) at least one $H_i$ is infinite.

In case (a), define $u_i := \max H_i$ for all $i \geq 0$.
Then the set $\lbrace u_0, u_1, \ldots \rbrace$ must be infinite,
for otherwise $F$ would be finite.  Choose an infinite subsequence
of the $u_i$ consisting of distinct elements, and apply Lemma~\ref{dick}.
Then there is an infinite increasing subsequence of indices
$i_1 < i_2 < \cdots$ 
such that $u_{i_1} \divides u_{i_2} \divides \cdots$.  So the 
sequence $(u_{i_j})_{j \geq 1}$ is strictly increasing.

Now consider the characteristic sequence of $S$, say
$(f(n))_{n \geq 0}$, taking values $1$ if $n \in S$ and $0$ otherwise.
Consider the subsequences $(f_j)$ in the $k$-kernel of
$f$ defined by $f_j (n) = f(k^{i_j} n)$ for $n \geq 0$, $j \geq 1$.
By our construction, the largest $n$ in $\pi(D)$ such that
$k \nodiv n$ and 
$f_j (n) = 1$ is $n = u_{i_j}$.   Since the $u_{i_j}$ are strictly
increasing, this shows the (infinitely many) sequences $(f_j)$
are pairwise distinct.  Hence, by Theorem~\ref{kernel},
$f$ is not $k$-automatic and neither is $S$.

In case (b),
we have $H_i$ is infinite for some $i$. 
As mentioned above, $S$ can be written as the disjoint union 
$\bigcup_{i \geq 0} k^i H_i$.
Let $L$ be the
language of canonical base-$k$ expansions of elements of $H_i$ (so that,
in particular, no element of $L$ starts with $0$).
The base-$k$ representation of elements of $k^i H_i$ end in
exactly $i$ $0$'s, and no other member of $S$ has this
property.  Since $S$ is assumed to be $k$-automatic, 
it follows that $L$
is regular.    Note that no two elements of $H_i$ have a quotient which 
is divisible by $k$, because if they did, the numerator would be divisible
by $k$, which is ruled out by the condition.

Since $L$ is infinite and regular,
by the pumping lemma, there must be words $u, v, w$, with
$v$ nonempty, such that $u v^j w \in L$ for all $j \geq 0$.  
Note that for all integers $j \geq 0$ and $c \geq 0$ we have
\begin{equation}
[u v^{j+c} w]_k = [u v^j w]_k \cdot k^{c|v|} + ([v^c w]_k - [w]_k \cdot
	k^{c|v|}).
\label{uvjc}
\end{equation}
Let $D = \lbrace p_1, p_2, \ldots, p_t \rbrace$.
Since $[u v^j w]_k \in S \subseteq \pi(D)$,
it follows that there exists a double sequence
$(f_{r,j})_{1 \leq r \leq t; j \geq 1}$ of non-negative integers such that 
\begin{equation}
[u v^j w]_k = p_1^{f_{1,j}} \cdots p_t^{f_{t,j}}
\label{factor}
\end{equation}
for all $j \geq 0$.
From (\ref{factor}), we see that
$k^{|uw|+j|v|} <  p_1^{f_{1,j}} \cdots p_t^{f_{t,j}}$,
and hence (assuming $p_1 < p_2 < \cdots < p_t$) we get
$$ (|uw| + j|v|) \log k < \left( \sum_{1 \leq r \leq t} f_{r,j} \right)
\log p_t.$$
Therefore, there are constants
$0 < c_1$ and $J$ such that
$ c_1 j < \sum_{1 \leq r \leq t} f_{r,j}$
for $j \geq J$.

For each $j\geq J$ now consider the indices $r$ such that $f_{r,j} > c_1 j/t$;
there must be at least one such index, for otherwise
$f_{r,j} \leq c_1 j/t$ for each $r$ 
and hence $\sum_{1 \leq r \leq t} f_{r,j} \leq c_1 j$, a contradiction.
Now consider $t+1$ consecutive $j$'s; for each $j$ there is an index $r$
with $f_{r,j} > c_1 j/t$, and since there are only $t$ possible indices,
there must be an $r$ and
two integers $l'$ and $l$, with $0 \leq l< l' \leq t$, such that
$f_{r,j+l} > c_1 (j+l)/t$ and $f_{r,j+l'} > c_1 (j+l')/t$.  This is true in any
block of $t+1$ consecutive $j$'s that are $\geq J$.  
Now there are infinitely many disjoint blocks of $t+1$ consecutive $j$'s,
and so there must be some pair $(r, l'-l)$ that occurs infinitely often.  Put $\delta = l' - l$.

Now use (\ref{uvjc}) and take $c = \delta$.  We get infinitely many $j$ such 
that 
$$ p_1^{f_{1,j+\delta}} \cdots p_t^{f_{t,j+\delta}} = k^{\delta|v|} p_1^{f_{1,j}} \cdots
	p_t^{f_{t,j}} + E,$$
where $E = [v^\delta w]_k - [w]_k \cdot k^{\delta|v|}$ is a constant that is
independent of $j$.
Now focus attention on the exponent of $p_r$ on both sides.
On the left it is $f_{r,j+\delta}$, which we know to be at least
$c_1 (j+\delta)/t$.  On the right the exponent of $p_r$ dividing the
first term is $f_{r,j} + \delta|v|e_r$ (where $k = p_1^{e_1} \cdots p_t^{e_t}$);
this is at least $c_1 j/t$.  So $p_r^h$ divides $E$, where
$h \geq c_1 j/t$.  But this quantity goes to $\infty$, and $E$ is a
constant.  So $E = 0$.  But then 
$$ {{[u v^{j+\delta} w]_k} \over {[uv^j w]_k}} = k^{\delta|v|} .$$
which is impossible, since, as we observed above, two elements of
$H_i$ cannot have a quotient which is a power of $k$.
This contradiction shows that $H_i$ cannot be infinite.

So now we know that $F$ is finite.  Fix some $f \in F$ and consider
$T_f = \lbrace k^i \, : \, k^i f \in S \rbrace$.  Since $S$ is $k$-automatic,
and the set of base-$k$
expansions $(T_f)_k$ is essentially
formed by stripping off the bits corresponding
to $(f)_k$ from the front of each element of $S$ of which
$(f)_k$ is a prefix, and replacing it with ``1'',
this is just a left quotient followed by concatenation, and 
hence $(T_f)_k$ is regular.  Let $M'$ be a DFA for $(T_f)_k$, and
consider an input of $1$ followed by
$l$ $0$'s for $l = 0, 1, \ldots$ in $M'$.
Evidently we eventually get into a cycle,
so this says that $U_f = \lbrace i \, : \, k^i f \in S \rbrace$ is
ultimately periodic.  Hence $S$ is $k$-finite.

This completes the proof of Theorem~\ref{dkf}.
\end{proof}
We can now prove Theorem~\ref{integers}.

\begin{proof}
One direction is easy, since if $S$ is $(\Enn,k)$-automatic, then
there is an automaton accepting $(S)_k$.  We can now easily modify this
automaton to accept all words over $(\Sigma_k^2)^*$ whose $\pi_2$ projection
represents the integer $1$ and whose $\pi_1$ projection is an element of
$(S)_k$.  Hence $S$ is $(\Quep, k)$-automatic.

Now assume $S \subseteq \Enn$ is $(\Quep,k)$-automatic.   
If $S$ is finite, then the result is clear, so assume $S$ is infinite.
Let $L$ be a regular language with $\quo_k(L) = S$.
Without loss of generality we may assume
every representation in $L$ is canonical; there are no leading $[0,0]$'s.
Furthermore, by first
intersecting with $L_{\not=0}$ we may assume that $L$
contains no representations of the integer $0$.  Finally, we can also
assume, without loss of generality, that no representation contains
\textit{trailing} occurrences of $[0,0]$,
for removing all trailing $[0,0]$'s from all words in $L$ preserves regularity,
and it does not change the set of numbers represented, as it has the effect
of dividing the numerator and denominator by the same power of $k$.
Since the words in
$L$ represent integers only, the denominator of every representation must
divide the numerator, and hence if the denominator is divisible by $k$,
the numerator must also be so divisible.  Hence removing trailing zeroes
\textit{also} ensures that no denominator is divisible by $k$.
Let $M$ be a DFA of $n$ states accepting $L$.

We first show that the set $[\pi_2(L)]_k$ of possible
denominators represented by $L$ is finite.
Write $S = S_1 \ \cup \ S_2$, where 
$S_1 = S \ \cap \ I[0, k^{n+1})$ and
$S_2 = S \ \cap \ I[k^{n+1}, \infty)$.  
Let $L_1 = L \ \cap \ L_{< k^{n+1}}$, the representations of all numbers
$< k^{n+1}$ in $L$, and
$L_2 = L \ \cap \ L_{\geq k^{n+1}}$.  Both $L_1$ and $L_2$ are regular,
by Lemma~\ref{lem1}.
It now suffices to show
that $S_2$ is $(\Enn,k)$-automatic.

Consider any $t \in S_2$. 
Let $z \in L_2$ be a 
representation of $t$.   Since $t \geq k^{n+1}$, clearly $|z| \geq n$,
and so 
$\pi_2(z)$ must begin with at least $n$
$0$'s.   Then, by the pumping lemma, we can write $z = uvw$ with $|uv| \leq n$
and $|v| \geq 1$ such that $uv^i w \in L$ for all $i \geq 0$.  However,
by the previous remark about $\pi_2(z)$, we see that $\pi_2(v) = 0^j$ for
$1 \leq j \leq n$. Hence 
$[\pi_2(z)]_k = [\pi_2(uvw)]_k = [\pi_2(uw)]_k$.
Since $uw$ must also represent a member of $S$, it must be an integer,
and hence
$[\pi_2(z)]_k \divides [\pi_1(uw)]_k$ as well as
$[\pi_2(z)]_k \divides [\pi_1(uvw)]_k$.  Hence
$$[\pi_2(z)]_k \divides [\pi_1(uvw)]_k - [\pi_1(uw)]_k
= ([\pi_1(uv)]_k - [\pi_1(u)]_k)\cdot k^{|w|}.$$

The previous reasoning applies to any $z \in L_2$.
Furthermore, $0 < [\pi_1(uv)]_k - [\pi_1(u)]_k  < k^n$.
It follows that every possible denominator $d$ of elements in $L_2$
can be expressed as $d = d_1 \cdot d_2$, where $1 \leq d_1 < k^n$ and
$d_2 \divides k^m$ for some $m$.   It follows that the set of primes
dividing all denominators $d$ is finite, and we can therefore apply
Theorem~\ref{dkf}.   Since $k$ divides no denominator, the set of
possible denominators is finite.

We can therefore decompose $L_2$ into a finite disjoint union corresponding
to each possible denominator $d$.  Next, we use a finite-state transducer to
divide the numerator and denominator of the corresponding representations
by $d$.  For each $d$, this gives a new regular language $A_d$ where the
denominator is $1$.  Writing $T := \bigcup_{d} A_d$, we have
$S_2 = \quo_k(T) = \bigcup_{d} \quo_k(A_d)$.  Now we project, throwing
away the second coordinate of elements of $T$; the result is regular and
hence $S$ is a $k$-automatic set of integers.
\end{proof}

\begin{corollary}
Let $L \subseteq (\Sigma_k^2)^*$ be a regular language of words
with no leading or trailing $[0,0]$'s, and suppose
$\quo_k (L) \subseteq \Enn$.     Then $L$ can (effectively)
be expressed as the
finite union
$$ \bigcup_{1 \leq i \leq A}
\lbrace (m a_i, m)_k \ : \ m \in S_i \rbrace \ \cup \ 
\bigcup_{1 \leq j \leq B} \lbrace (b_j n, b_j)_k \ : \ n \in T_j \rbrace,$$
where $A, B \geq 0$ are integers, and $S_1, S_2, \ldots, S_A$
and $T_1, T_2, \ldots T_B$ are $(\Enn, k)$-automatic sets of integers,
and $a_1, a_2, \ldots, a_A$ and $b_1, b_2, \ldots, b_B$ are 
non-negative integers.
\label{deco}
\end{corollary}

\begin{proof}
Following the proof of Theorem~\ref{integers}, we see that $L$ can be written
as the union of 
$L_1$, the representations of integers $<k^{n+1}$ and
$L_2$, the representations of those
$ \geq k^{n+1}$.   For each integer $ a < k^{n+1}$ we can consider the
words of $L$ whose quotient gives $a$; this provides a partition of $L_1$ into
regular subsets corresponding to each $a$, and gives the first term
of the union.
For $L_2$, the proof of Theorem~\ref{integers} shows that there are only
a finite number of possible denominators, and the sets of corresponding
numerators are $k$-automatic.  
\end{proof}

As corollaries, we get that the $k$-automatic sets of rationals are
(in contrast with sets of integers) not necessarily closed under the
operations of intersection and complement.

\begin{theorem}
Let $S_1 = \lbrace (k^n-1)/(k^m-1) \, : \, 1 \leq m < n \rbrace$
and $S_2 = \Enn$.  Then $S_1$ and $S_2$ are both $k$-automatic sets of
rationals, but $S_1 \ \cap \ S_2$ is not.
\label{inter}
\end{theorem}

\begin{proof}
We can write every element of $S_1$ as $p/q$, where $p = (k^n-1)/(k-1)$ and
$q = (k^m-1)/(k-1)$.  The base-$k$ representation of $p$ is
$1^n$ and the base-$k$ representation of $q$ is $1^m$.  Thus
a representation for $S_1$ is given by the regular
expression $[1,0]^+[1,1]^+$.    We know that $\Enn$ is $k$-automatic
from Example~\ref{example2}.  

From a classical result we know that $(k^m -1) \divides (k^n-1)$ if and
only if $m \divides n$.  It follows that $S_1 \ \cap \ S_2 = T$, where
$$T = \lbrace (k^n-1)/(k^m-1) \, : \, 1 < m < n
\text{ and } m \divides n \rbrace.$$
If the $k$-automatic sets of rationals were closed under intersection,
then $T$ would be $(\Enn,k)$-automatic.
Writing $ n = md$, 
we have 
$${{k^n-1}\over {k^m-1}} = k^{(d-1)m} + \cdots + k^m + 1,$$ whose
base-$k$ representation is $(10^{m-1})^{d-1} 1$.  
Hence $(T)_k = \lbrace (10^{m-1})^{d-1} 1 \, : \, m \geq 1, d > 1 \rbrace$.
Assume this is regular.
Intersecting with the regular language $10^*1 0^*1 0^*1$ we get
$\lbrace 1 0^n 1 0^n 1 0^n 1 \, : \, n \geq 1 \rbrace$.  But a routine
argument using the pumping lemma shows this is not even
context-free, a contradiction.
\end{proof}

From this result we can obtain several corollaries of interest.

\begin{corollary}
The class of $(\Quep,k)$-automatic sets is not closed under the operations
of intersection or complement.
\label{notclosed}
\end{corollary}

\begin{proof}
We have just shown that this class is not closed under intersection.
But since it is closed under union, if it were closed under complement,
too, it would be closed under intersection, a contradiction.
\end{proof}

\begin{corollary}
Define a ``normalization operation'' $N$ that maps
a word $w$ to its canonical expansion in lowest terms
$(p/d,q/d)_k$, where $d = \gcd(p,q)$ and
$p = \pi_1 (w)$, $q = \pi_2 (w)$, and define
$N(L) = \lbrace N(w) \ : \ w \in L \rbrace$.  
Then the operation $N$ does not, in general, preserve regularity.
\label{lowestterms}
\end{corollary}

\begin{proof}
We give an example of a regular language $L$ where $N(L)$ is not even
context-free.    It suffices to take
$L = (S_1)_k$, where $S_1$ is the set defined above in
the statement of Theorem~\ref{inter}.  Consider $N(L)$; then
we know from the argument above that
$S_1 \ \cap \ \Enn = T$, where
$$T = \lbrace (k^n-1)/(k^m-1) \, : \, 1 < m < n
\text{ and } m \divides n \rbrace.$$
Hence $N(L) \ \cap \ (\Enn)_k = (T)_k$, but from the argument above
we know that $(T)_k = 
\lbrace (10^{m-1})^{d-1} 1 \, : \, m \geq 1, d > 1 \rbrace$
is not context-free.
\end{proof}

\begin{corollary}
There is a $k$-automatic set of rationals whose base-$k$ expansions
(as real numbers) are not accepted by any B\"uchi automaton.
\label{notbuchi}
\end{corollary}

\begin{proof}
Define $S_3 = 
\lbrace (k^m-1)/(k^n-1) \, : \, 1 \leq m < n \rbrace$; this is 
easily seen to be a $k$-automatic set of rationals.
However, the set of its base-$k$ expansions is of the form
$$ \bigcup_{0 < m < n < \infty} 0.(0^{n-m} \, (k-1)^m)^\omega ,$$
where by $x^\omega$ we mean the infinite word $xxx\cdots$.
A simple argument using the pumping lemma shows that no
B\"uchi automaton can accept this language.
\end{proof}

It follows that the class of $k$-automatic sets of rational numbers we
study in this paper is not the same as that studied in
\cite{Adamczewski&Bell:2011,Boigelot&Brusten&Bruyere:2010,Boigelot&Brusten:2009,Boigelot&Brusten&Leroux:2009}.

\begin{remark}
The technique above also allows us to prove that the languages
$L_d =  \lbrace (p,q)_k \, : \, q \divides p \rbrace$,
$L_r = \lbrace (p,q)_k \, : \, \gcd(p,q) > 1 \rbrace$,
and
$L_g = \lbrace (p,q)_k \, : \, \gcd(p,q) = 1 \rbrace$
are not context-free.

For the first, suppose $L_d$ is context-free and is
accepted by a PDA $M_d$.
Consider a PDA $M$ that
on input a unary word $x := 1^n$, $n \geq 2$, guesses
a word of the form $y := 0^{n-a} 1^a$ with $1 < a < n$ and simulates
$M_d$ on $x \times y$, accepting if and only if
$x \times y \in L_d$.  Then $M$ accepts 
the unary language $\lbrace 1^n \ : \ n \text{ composite} \rbrace$
which is well-known to be non-context-free
\cite[Ex.\ 6.1, p.\ 141]{Hopcroft&Ullman:1979}, a contradiction.

Notice that $L_d$ can be considered as the set of all possible
rational representations of $\Enn$.

For $L_r$, the same kind of construction works.

For $L_g$, a more complicated argument is needed.  First, we prove
that the language $L_c = \lbrace 0^i 1^j \ : \ i, j \geq 1 
\text{ and } \gcd(i,j) > 1 \rbrace$ is not context-free.
To see this, assume it is,
use the pumping lemma, let $n$ be the constant, and
let $p$ be a prime $> n$.
Choose $z = 0^{p^2} 1^p \in L_c$.  Then 
we can write $z = uvwxy$ where $|vwx| \leq n$ and $|vx| \geq 1$
and $u v^i w x^i y \in L_c$ for all $i \geq 0$.
Then $v$ and $x$ each contain only one type of letter, because
otherwise $uv^2wx^2y$ has $1$'s before $0$'s, a contradiction.
There are three possibilities:  
\begin{itemize}
\item[(a)] $vx = 0^r$, $1 \leq r \leq n$;
\item[(b)] $vx = 1^s$, $1 \leq s \leq n$;
\item[(c)] $v = 0^r$ and $x = 0^s$, $1 \leq r,s \leq n$.
\end{itemize}

In case (a) we pump with $i = 2$, obtaining $u v^2 x w^2 y = 0^{p^2+r} 1^p$.
But $p$ is a prime, so for this word to be in $L_c$ we must have
$r \equiv \lmod{0} {p}$, contradicting the
inequality $1 \leq r \leq n < p$.

Similarly, in case (b) we get $0^{p^2} 1^{p+s}$.  
But $1 \leq s < p$, so $p < p+s < 2p$.
Hence $p$ does not divide $p+s$, so $\gcd(p^2, p+s) = 1$, a contradiction.

Finally, in case (c) and pumping with $i = j+1$ we get $0^{p^2+rj} 1^{p+sj}$.
Since $p$ is a prime, and $1 \leq r, s < p$, by Dirichlet's theorem
there are infinitely many primes of the form $p^2 + rj$ and $p + sj$.
If $r \geq s$, choose $j$ so that $p^2 + rj$ is a prime.  Since
$p^2 + rj > p+sj$, we have $\gcd(p^2+rj, p+sj) = 1$, a contradiction.

If $r < s$, choose $j > p^2-p$ such that $p+sj$ is prime.  Then
$p+ j > p^2$, so by adding $(s-1)j$ to both sides we get
$p + sj > p^2 + (s-1)j \geq p^2 + rj$.  Thus $p+sj$ is a prime greater
than $p^2 + rj$ and so $\gcd(p^2+rj, p+sj) = 1$, a contradiction.
This completes the proof.

From this we also get that 
$$L_e = \lbrace 0^i 1^j \ : \ \gcd(i,j) = 1 \rbrace$$ 
is not context-free, since it is known that the class of
context-free languages that are subsets of $0^* 1^*$ are closed under
relative complement with $0^* 1^*$ \cite{Ginsburg:1966}.

We can now prove that
$L_g = \lbrace (p,q)_k \, : \, \gcd(p,q) = 1 \rbrace$ is not 
context-free.      Suppose it is.
Then, since the CFL's are closed under
intersection with a regular language, $L' := L_g \ \bigcap \ 
([1,0]^*[1,1]^+ \ \cup \ [0,1]^*[1,1]^+)$ is also context-free.
But the numerators are numbers of the
form $(k^n-1)/(k-1)$ and the denominators are numbers of the form
$(k^m - 1)/(k-1)$.  From a classical result, we know that
$\gcd(k^m - 1, k^n - 1) = k^{\gcd(m,n)} - 1$.  So it follows that
$$L' =
\lbrace [1,0]^s [1,1]^t \ : \ \gcd(s+t,t) = 1 \text{ and } s\geq 0, t \geq 1 \rbrace $$
$$
\ \cup \ \lbrace [0,1]^s [1,1]^t \ : \ \gcd(s+t,t) = 1 \text{ and } s\geq 0, t \geq 1 \rbrace.$$
Now apply the morphism $h$ that maps $[a,b]$ to $a+b-1$.  Since the CFL's are
closed under morphism, we get
$$h(L') = \lbrace 0^m 1^n \ : \ \gcd(m,n) = 1 \text{ and } m, n \geq 1  \rbrace.$$
But we already proved this is not context-free, 
a contradiction.  Hence $L_g$ is not context-free.
\label{rem1}
\end{remark}

\section{Solvability results}
\label{solvability}

In this section we show that a number of problems involving $k$-automatic
sets of integers and rational numbers are recursively solvable.

\begin{theorem}
The following problems are recursively solvable:  given
a DFA $M$, a rational number $\alpha$, and a relation
$\rela$ chosen from $=, \not=, <, \leq, >, \geq $, does
there exist $x \in \quo_k(L(M))$ with $x \rela \alpha$?
\end{theorem}

\begin{proof}
The following gives a procedure for deciding if
$x \rela \alpha$.  First, we create a DFA $M'$ accepting the
language $L_{\rela \alpha}$ as described in Lemma~\ref{lem1} above.
Next, using the usual direct product construction, we create
a DFA $M''$ accepting $L(M) \ \cap \ L_{\rela \alpha}$.
Then, using breadth-first or depth-first search, we check to see
whether there exists
a path from the initial state of $M''$ to some final
state of $M''$.  
Since by definition $L_{\rela \alpha}$ contains every representation of each rational $x \in \quo_k(L_{\rela \alpha})$, we have $\quo_k(L(M) \ \cap \ L_{\rela \alpha}) = \quo_k(L(M)) \ \cap \ \quo_k(L_{\rela \alpha})$, and therefore this procedure is correct.
\end{proof}

\begin{lemma}
Let $M$ be a DFA with input alphabet
$\Sigma_k^2$ and let $F \subseteq \Quep$ be a finite set of
non-negative rational numbers.
Then the following problems are recursively solvable:
\begin{enumerate}
\item Is $F \subseteq \quo_k (L(M))$?
\item Is $\quo_k(L(M)) \subseteq F$?
\end{enumerate}
\label{lemma1}
\end{lemma}

\begin{proof}
To decide if $F \subseteq \quo_k (L(M))$, we simply check, using
Lemma~\ref{lem1}, whether $x \in \quo_k (L(M))$ for each $x \in F$.

To decide if $\quo_k(L(M)) \subseteq F$, we create DFA's accepting
$L_{=x}$ for each $x \in F$, using Lemma~\ref{lem1}.  Now we
create an automaton accepting all representations of all elements of $F$
using the usual direct product construction for the
union of regular languages.
Since $F$ is finite, the resulting automaton $A$ is finite.
Now, using the usual
direct product construction, we create a DFA accepting
$L(M)-L(A)$ and check to see if its language is empty.
\end{proof}

\begin{theorem}
The following problem is recursively solvable:  given a 
DFA $M$, and an integer $k$,
is the set $\quo_k(L(M))$ infinite?
\label{infinitec}
\end{theorem}

Note that this is \textit{not} the same as asking whether the language
$L(M)$ itself is infinite, since a single number may have infinitely many 
representations.  

First, we need 
a useful lemma from
\cite{Shallit:2011,Schaeffer&Shallit:2012}.

\begin{lemma}
Let $u, v, w \in (\Sigma_k^2)^*$ such that
$|v| \geq 1$, and such that
$[\pi_1(uvw)]_k$ and
$[\pi_2(uvw)]_k$ are not both $0$.  
Define 
\begin{equation}
U := \begin{cases}
\quo_k(w), & \text{if $[\pi_1(uv)]_k = [\pi_2(uv)]_k = 0$}; \\
\infty, & \text{if $[\pi_1(uv)]_k > 0$ and $[\pi_2(uv)]_k = 0$}; \\
{{[\pi_1(uv)]_k - [\pi_1(u)]_k} \over
	{[\pi_2(uv)]_k - [\pi_2(u)]_k}},
	& \text{otherwise.}
\end{cases}
\end{equation}
(a) Then exactly one of the following cases occurs:
\begin{itemize}
\item[(i)] $\quo_k(uw) < \quo_k(uvw) < \quo_k(uv^2w) < \cdots < U$ ;
\item[(ii)] $\quo_k(uw) = \quo_k(uvw) = \quo_k(uv^2w) = \cdots = U$ ;
\item[(iii)] $\quo_k(uw) > \quo_k(uvw) > \quo_k(uv^2w) > \cdots > U$ .
\end{itemize}
(b) Furthermore, $\lim_{i \rightarrow \infty} \quo_k(uv^iw) = U$.
\label{mainl}
\end{lemma}

Now we can prove Theorem~\ref{infinitec}.

\begin{proof}
Without loss of generality, we may assume that the representations
in $M$ are canonical (contain no leading $[0,0]$'s).
Define 
$$\gamma_k (u,v) = {{[\pi_1(uv)]_k - [\pi_1(u)]_k} \over
	{[\pi_2(uv)]_k - [\pi_2(u)]_k}},$$
and let $\pref(L)$ denote the language of all prefixes of all words
of $L$.   Let $n$ be the number of states in $M$.
We claim that the set $\quo_k(L(M))$ is finite if and only if
$\quo_k(L(M)) \subseteq T$, where
\begin{multline}
T = \lbrace \quo_k (x) \, : \, x \in L(M) \text{ and } |x| < n \rbrace
\ \cup \\
\lbrace \gamma_k(u,v) \, : \, uv \in \pref(L) \text{ and }
	|v| \geq 1 \text{ and } |uv| \leq n \rbrace .
\end{multline}

One direction is easy, since if $\quo_k(L(M)) \subseteq T$, then
clearly the set $\quo_k(L(M))$ is finite, since $T$ is.

Now suppose $\quo_k(L(M)) \subsetneq T$, so there exists
some $x \in L(M)$ with $\quo_k(x) \not\in T$.
Since $T$ contains all words of $L(M)$ of length $<n$, 
such an $x$ is of length $\geq n$.
So the pumping lemma applies, and there exists a decomposition
$x = uvw$ with $|uv| \leq n$ and $|v| \geq 1$ such that
$u v^i w \in L$ for all $i \geq 0$.  
Now apply Lemma~\ref{mainl}.  If case (ii) of that
lemma applies, then $\quo_k (x) = \gamma_k(u,v) \in T$, a 
contradiction.
Hence either case (i) or case (iii) must apply, and the lemma
shows that $\quo_k (u v^i w)$ for $i \geq 0$ gives
infinitely many distinct elements of $\quo_k(L(M))$.

To solve the decision problem, we can now simply enumerate the
elements of $T$ and use Lemma~\ref{lemma1}.
\end{proof}

\begin{theorem}
Given $p/q \in \Quep$, and a DFA $M$ accepting a $k$-automatic set
of rationals $S$, it is decidable if $p/q$ is an accumulation point
of $S$.
\end{theorem}

\begin{proof}
The number $\alpha$ is an accumulation point of a set of
real numbers $S$ if and only if at least one of the following
two conditions holds:
\begin{itemize}
\item[(i)] $\alpha = \sup (\, S \ \cap \ I(-\infty, \alpha) \, )$;
\item[(ii)] $\alpha = \inf (\, S \ \cap \ I(\alpha, \infty) \, )$.
\end{itemize}
By Lemma~\ref{lem1} we can compute a DFA accepting
$S' := S \ \cap \ I(-\infty, \alpha)$ 
(resp., $S \ \cap \ I(\alpha, \infty)$).
By \cite[Thm.\ 2]{Shallit:2011}
we can compute $\sup S'$ (resp., $\inf S'$).  
\end{proof}

\begin{theorem}
Suppose $S$ is a $k$-automatic set of integers accepted by
a finite automaton $M$.  There is an algorithm to decide, given $M$, 
whether there exists a finite set $D \subseteq \primes$ such that
$S \subseteq \pi(D)$.  Furthermore, if such a $D$
exists, we can also determine the sets $F$ and $U_f$ in
Theorem~\ref{dkf}.
\label{decide1}
\end{theorem}

\begin{proof} 
To determine if such a $D$ exists, it suffices to remove all trailing
zeroes from words in $(S)_k$ and see if the resulting language is finite.
If it is, we know $F$, and then it is a simple matter to compute the
$U_f$. 
\end{proof}

\begin{theorem}
There is an algorithm that, given a DFA $M_1$ accepting
$L_1 \subseteq (\Sigma_k^2)^*$, will determine if $\quo_k(L_1) \subseteq
\Enn$.  If so, the algorithm produces an automaton $M_2$ such that
$[L(M_2)]_k = \quo_k (L_1)$.  
\label{algo}
\end{theorem}

\begin{proof}
We start with the algorithm to determine if $\quo_k(L_1) \subseteq
\Enn$.

Modify $M_1$, if necessary, to accept only representations of nonzero
numbers, and to accept the remaining
words of $L_1$ stripped of
leading and trailing $[0,0]$'s.  Let $M_1$ have $n$ states.

1.  Create, 
using Lemma~\ref{lem1}, a 
DFA $M_3$ accepting a representation of the set
$T_1 := (\quo_k(L_1) \ \cap \ I(0,k^{n+1}) ) \setminus 
\lbrace 0, 1, \ldots, k^{n+1} - 1 \rbrace$.  If $T_1 \not= \emptyset$,
then answer ``no'' and stop.

2. Next, create a DFA $M_4$ accepting the
set $T_2 := (\quo_k(L_1) \ \cap \ I[k^{n+1}, \infty) )$.  
If any denominator ends in $0$, answer ``no'' and stop.

3.  Compute $\pi_2 (L(M_4))$ and,
using Theorem~\ref{decide1}, decide if the
integers so represented are factorable into a finite set of primes.
If not, answer ``no'' and stop.

4.  Otherwise, compute the decomposition in Corollary~\ref{deco}, obtaining 
the finite set of denominators in that decomposition.  Check
whether each denominator divides all of the corresponding
numerators.  If not, answer ``no'' and stop.  Otherwise, answer
``yes''.

To see that this works, note that if some non-integer $\alpha$ belongs
to $\quo_k (L_1)$, then either $\alpha < k^{n+1}$ or $\alpha \geq k^{n+1}$.
In the former case we have $T_1 \not=\emptyset$, so this will be detected
in step 1.  

Otherwise $\alpha \geq k^{n+1}$.  Then either the set of denominators
of $L_1$ do not factor into a finite set of primes (which will be
detected in step 3), or they do.  In the latter case the set of
denominators is $k$-finite, by Theorem~\ref{dkf} and so has a representation
in the form given by Lemma~\ref{kfinite} (b).  Now $k$ cannot divide both a
numerator and denominator, because we have removed trailing $[0,0]$s from
every representation.   So $k$ divides a numerator but not a denominator
if and only if this is detected in step 2.
If steps 2 and 3 succeed, then, there are only finitely many denominators.

Now $\quo_k (L_1) \subseteq \Enn$ if and only if the numerators $n$
corresponding to each of these finitely many denominators $d$ are
actually divisible by $d$.  We can (effectively)
form a partition of $L_1$ according
to each denominator.  Using a machine to test divisibility by $d$,
we can then intersect with each corresponding machine in the partition
to ensure each numerator is indeed divisible.   If so,
we can easily produce an $(\Enn, k)$-automaton accepting $\quo_k (L_1)$.
\end{proof}

\begin{corollary}
There is an algorithm that, given a DFA $M_1$ accepting
$L_1 \subseteq (\Sigma_k^2)^*$ and a DFA $M_2$ accepting
$L_2 \subseteq (\Sigma_k)^*$, will decide
\begin{itemize}
\item[(a)] if $\quo_k(L_1) \subseteq [L_2]_k$;
\item[(b)] if $\quo_k(L_1) = [L_2]_k$.
\end{itemize}
\end{corollary}

\begin{proof}
(a)  The algorithm is as follows:
using the algorithm in the proof of Theorem~\ref{algo}, first
determine if $\quo_k(L_1) \subseteq \Enn$.   If not, answer ``no''.
If so, using the algorithm in that proof,
we determine an automaton $M$ such that $\quo_k (L_1) = [L(M)]_k$.
Finally, using the usual cross-product construction, we create
an automaton $M'$ that accepts
$L(M) \setminus L_2$.  If $M'$ accepts anything, then answer ``no'';
otherwise answer ``yes''.

(b) Similar to the previous case.  In the last step, we create an
automaton $M'$ that accepts the symmetric difference
$(L(M) \setminus L_2) \ \bigcup \ (L_2 \setminus L(M))$.
\end{proof}

In particular, it is decidable if $\quo_k(L_1) = \Enn$.

\section{Subword complexity of denominators}

The {\it subword complexity} of a language $L \subseteq \Sigma^*$ is the function
$f_L: \Enn \rightarrow \Enn$ defined by
$f_L(n) = | \Sigma^n \ \cap \ L |$,
the number of distinct words
of length $n$ in $L$.  

The following is a natural question;  suppose $L$ is a language
such that $\quo_k (L) = \Quep$.  What is the smallest possible
subword complexity of the denominators ${\pi_2(L)}$?
If no further restrictions on $L$ are given, then
it is easy to see that $f_{\pi_2(L)}$ can grow arbitrarily slowly (say,
by enumerating the rational numbers and then finding arbitrarily large
representations for each one).  However, if $L$ is restricted to be
regular, then the best we can do is quadratic, as the following two
results show.

\begin{theorem}
If $L$ is a regular language such that $\quo_k(L) = \Quep$, then
$\pi_2(L)$ (resp., $\pi_1(L)$) is not of subword complexity 
$\, o(n^2) \, $.
\label{dd1}
\end{theorem}

\begin{proof}
We prove the result for $\pi_2$; an analogous proof works for $\pi_1$.

Suppose there exists a regular language $L$ with
$\pi_2(L) = o(n^2)$.  Then by a theorem
of Szilard et al.\ \cite{Szilard&Yu&Zhang&Shallit:1992}, we know
that the subword complexity of $\pi_2(L)$ must be $O(n)$.  By another theorem in
that same paper, we know that this implies that
we can write $\pi_2(L)$ as the finite union 
$$ \pi_2(L) = \bigcup_{1 \leq i \leq n} u_i v_i^* w_i x_i^* y_i $$
where the $u_i, v_i, w_i, x_i, y_i$ are (possibly empty) finite words.

Suppose every
$u_i v_i$, for $1 \leq i \leq n$,
contains a nonzero symbol.  Then for every word $z \in L$ we
would have $\quo_k (z) < k^M$, where $M = \max_{1 \leq i \leq n}
|u_i v_i|$, and hence we could not represent arbitrarily large rational
numbers, a contradiction.  It follows that there must be some nonempty subset
$\cal S$ of the indices $\lbrace 1,2,\ldots, n \rbrace$ 
such that $u_i v_i \in 0^*$ for all $i \in \cal S$.   

Similarly, suppose every $x_i y_i$, for $i \in \cal S$, contains
a nonzero symbol.  Then for every word $z \in L$ with
$\pi_2(z) \in \bigcup_{i \in \cal S} u_i v_i^* w_i x_i^* y_i$
we would have $\nu_k(\pi_2(z)) <
N$, where $N = \max_{1 \leq i \leq n} |x_i y_i|$.  But then we could
not represent all rational numbers of the form $(kp+1)/(q k^j)$, where
$j > N$ and $kp + 1 > k^{M+j} q$, a contradiction.  It follows that
there must be some nonempty subset ${\cal S}'$ of $\cal S$ such that
both $u_i v_i$ and $x_i y_i$ are in $0^*$ for all $i \in {\cal S}'$.  

From the above argument, all rational numbers 
of the form $(kp+1)/(q k^j)$ with
$j > N$ and $kp + 1 > k^{M+j} q$ must be represented
by $z \in L$ 
with $\pi_2(z) \in \bigcup_{i \in {\cal S}'} u_i v_i^* w_i x_i^* y_i$.  Since
$u_i v_i$ and $x_i y_i$ are in $0^*$ for each such term, it follows
that the set $\cal T$ of all prime factors of denominators of all these words
is finite (and consists of the prime factors of $k$ and $\pi_2 (w_i)$
for $i \in {\cal S}'$).  Choose any prime $r \not\in \cal T$  
and consider the rational number $r'/(r k^N)$, where $r'$ is any
prime with $r' > r k^{M+N}$.  Then this number has no representation,
a contradiction.
\end{proof}

\begin{theorem}
For each $k \geq 2$,
there exists a regular language $L$ such that $\quo_k (L) = \Quep$,
and $f_{\pi_2(L)} (n) = \Theta(n^2)$.
\label{dd2}
\end{theorem}

\begin{proof}
Let
$$L_2 = 
\lbrace w \in (\Sigma_k^2)^* \ : \ \pi_2 (w) \in 0^* 1^+ 0^* \rbrace $$
be the language given above in Example~\ref{six}.  We claim that
there are exactly $n(n+1)/2$ words of
length $n$ in $\pi_2(L_2)$.  
To see this, count words of the form
$0^* 1^+ 0^*$ of length $n$.
There is $1$ word consisting of
all $1$'s, $2$ words with $n-1$ consecutive $1$'s, $3$ words
with $n-2$ consecutive $1$'s, and so forth, for a total
of $1 + 2 + \cdots + n = n(n+1)/2$ words.
\end{proof}

\section{Open problems}
\label{open}

     There are a number of open problems raised by this work.  The
most outstanding one is a generalization of Cobham's theorem \cite{Cobham:1969} to the setting of $k$-automatic sets of rationals.

     Let $r \geq 1$ be an integer.
We say a set $A \subseteq \Enn^r$ is {\it linear} if there
exist vectors $v_0, v_1, \ldots, v_i \in \Enn^r$ such that
$$ A = \lbrace v_0 + a_1 v_1 + \cdots + a_i v_i \ : \ a_1, a_2, \ldots,
a_i \in \Enn \rbrace.$$ 
We say a set is {\it semilinear} if it is the finite union of linear
sets.

      Given a subset $A \subseteq \Enn^2$,
we can define its set of quotients $q(A)$ to be $\lbrace p/q \ : \
[p,q] \in A \rbrace$.

\begin{conjecture}
$S$ is a set of rational numbers that is simultaneously
$k$- and $l$-automatic for multiplicatively independent integers
$k, l \geq 2$ if and only if
there exists a semilinear set $A \subseteq \Enn^2$ such that
$S = q(A)$.
\end{conjecture}

One direction of this conjecture is clear, as given $A$ we can easily
build an automaton to accept the base-$k$ representation of the
set of rationals $q(A)$.  The converse, however, is
not so clear.

     We now turn to some decision problems whose status we have not
been able to resolve.  Which of the following problems, if any, are
recursively solvable?

Given a DFA accepting
$L \subseteq (\Sigma_k^2)^*$ representing a $k$-automatic
set of rationals $S$,

\begin{enumerate}

\item  does $S$ contain at least one integer?  (alternatively:
is there $(p,q)_k \in L$ such that $q \divides p$?)

\item does $S$ contain infinitely many integers?  

\item are there infinitely many $(p,q)_k \in L$ such that
$q \divides p$?

\item is there some rational number $p/q \in S$ having infinitely
many distinct representations in $L$?

\item are there infinitely many distinct rational numbers $p/q \in S$
having infinitely many distinct representations in $L$?

\item do all rational numbers $p/q \in S$ have infinitely many 
distinct representations in $L$?

\item do all the $(p,q)_k \in L$ have $\gcd(p,q) = 1$?

\item does any $(p,q)_k \in L$ have $\gcd(p,q) = 1$?

\item is $\max_{(p,q)_k \in L} \gcd(p,q)$ bounded?

\end{enumerate}

In the particular case where $S = \Quep$, we have the following conjecture.

\begin{conjecture}
If $L$ is a regular language with $\quo_k (L) = \Quep$, then 
$L$ contains infinitely many distinct representations for infinitely
many distinct rational numbers.
\end{conjecture}

We note that it is possible to construct a regular $L$ with
$\quo_k(L) = \Quep$ where there are infinitely many rational numbers
with exactly one representation.  For example, if $k = 2$, we can take the
language $L_3$ of Example~\ref{three} and add back a single representation for
each element of the set $\lbrace 2^n \ : \ n \in \Zee \rbrace$.

Here are some additional decision problems whose status (recursively
solvable or unsolvable) we have not been able to resolve:

Given DFA's $M_1$ and $M_2$ accepting languages $L_1, L_2 \subseteq
(\Sigma_k^2)^*$ representing sets $S_1, S_2 \subseteq \Quep$,

\begin{enumerate}

\item is $S_1 = S_2$?

\item is $S_1 \subseteq S_2$ ?

\item is $S_1 \ \cap \ S_2 \not= \emptyset$?

\end{enumerate}

We hope to address some of these questions in a future paper.

It has long been known that 
$\Th(\Enn, +, |)$, where $|$ denotes the divisibility relation,
is undecidable; see, for example, \cite[pp.\ 78--79]{Tarski:1953}.
We observe that the decidability of assertions involving even just
two or three quantifiers, divisibility, and automata,
would allow the solution of two classic
open problems from number theory.

\begin{example}
Consider the language
$L_7 \subseteq (\Sigma_2^2)^*$ defined by words with first
component representing a numerator $n$ of the form $2^i + 1$ for $i > 32$ and
denominator an odd number $d$  with $1 < d < n$.  
This is clearly regular.  Now consider the assertion
$$\exists p\  \forall q  \ (q \mid p) \implies (p,q)_2 \not\in L_7 . $$
This assertion is true if and only if there exists a Fermat prime
greater than $2^{32} + 1$.
\end{example}

\begin{example}
Consider the language
$L_8 \subseteq (\Sigma_2^2)^*$ defined by words with first
component representing a numerator $n$ of the form $2^i - 1$ for $i \geq 3$,
and denominator an odd number $d$ with $1 < d < n$.
This is easily seen to be regular.  
Now consider the assertion
$$ \forall t \ \exists p > t \ \forall q \ (q \mid p) \implies (p,q)_2 \not\in L_8 .$$
This assertion is true if and only if there are infinitely many Mersenne primes.
\end{example}

\section{Acknowledgments}

We thank \'Emilie Charlier, Luke Schaeffer,
Jason Bell, Narad Rampersad, Jean-Paul Allouche, and Thomas
Stoll for their helpful comments.

\end{document}